\documentclass[11pt,notitlepage]{article}
\usepackage{microtype}
\usepackage[T1]{fontenc}
\usepackage[utf8]{inputenc}
\usepackage[english]{babel}
\usepackage{algorithmic}
\usepackage[ruled,vlined,linesnumbered]{algorithm2e}
\usepackage{enumitem}
\setlist{nosep} 

\usepackage{csquotes}
\usepackage{graphicx}
\usepackage{xcolor}
\usepackage{authblk}
\usepackage{todonotes}
\usepackage{slashbox}

\usepackage{commath}
\usepackage[sorting=none,maxbibnames=99]{biblatex}
\usepackage{float}
\usepackage{url}
\usepackage[colorlinks=false, citecolor=blue]{hyperref}
\bibliography{main.bib}
\usepackage[normalem]{ulem}
\usepackage{xspace,geometry,fancyhdr}
\usepackage{setspace}
\usepackage{caption}
\usepackage{subcaption}
\usepackage{amssymb,amsmath,amsthm}
\newtheorem{theorem}{Theorem}
\usepackage{mathtools} 
\usepackage{thmtools}
\usepackage{empheq} 
\newtheorem{lemma}[theorem]{Lemma}

\newtheorem{claim}[theorem]{Claim}
\newtheorem{observation}[theorem]{Observation}
\newtheorem{corollary}[theorem]{Corollary}

\theoremstyle{remark}

\usepackage{booktabs}
\usepackage{algorithmic}
\usepackage{multirow}
\usepackage{siunitx}
\usepackage[toc,title,page]{appendix}

\newenvironment{proofof}[1]{\begin{proof}[Proof of #1]}{\end{proof}}

\usepackage{fancyhdr}

\usepackage{calc}
\usepackage{changepage}
\usepackage{glossaries}
\glsdisablehyper
\newacronym{ML}{ML}{Machine Learning}
\newacronym{LHS}{LHS}{Left-Hand-Side}
\newacronym{RHS}{RHS}{Right-Hand-Side}
\newacronym{SRPT}{SRPT}{Shortest Remaining Processing Time}
\newacronym{DBP}{DBP}{Dynamic Bin Packing}
\newacronym{CNN}{CNN}{Convolutional Neural Network}
\newacronym{FCFS}{FCFS}{First Comes First Served}
\newacronym{SPT}{SPT}{Shortest Processing Time}
\newacronym{SVF}{SVF}{Smallest Volume First}
\newacronym{WSPT}{WSPT}{Weighted Shortest Processing Time}
\newacronym{WSVF}{WSVF}{Weighted Smallest Volume First}
\newacronym{SJF}{SJF}{Shortest Job First}
\newacronym{VM}{VM}{Virtual Machine}
\newacronym{AWS}{AWS}{Amazon Web Services}
\newacronym{GCP}{GCP}{Google Cloud Platform}
\newacronym{OCO}{OCO}{Online Convex Optimization}
\newacronym{HRDF}{HRDF}{Highest Residual Density First}
\newacronym{HPCS}{HPCS}{High Priority Category Scheduling}

\usepackage{titling}
\usepackage{blindtext}

\title{Weighted completion time minimization for capacitated parallel machines}
\author{Ilan Reuven Cohen\thanks{Faculty of Engineering, Bar-Ilan University, Israel.{\tt ilan-reuven.cohen@biu.ac.il}}
	\and {Izack Cohen\thanks{Faculty of Engineering, Bar-Ilan University, Israel. {\tt izack.cohen@biu.ac.i}
	}}
	\and {Iyar Zaks\thanks{Faculty of Industrial Engineering and Management, The Technion - IIT, Israel. {\tt iyarzaks@gmail.com}
	}}
}

\begin{document}
	\maketitle
	\begin{abstract}
		We consider the weighted completion time minimization problem for capacitated parallel machines, which is a fundamental problem in modern cloud computing environments. We study settings in which the processed jobs may have varying duration, resource requirements and importance (weight). Each server (machine) can process multiple concurrent jobs up to its capacity.
		Due to the problem's $\mathcal{NP}$-hardness, we study heuristic approaches with provable approximation guarantees. We first analyze an algorithm that prioritizes the jobs with the smallest volume-by-weight ratio. We bound its approximation ratio with a decreasing function of the ratio between the highest resource demand of any job to the server's capacity.
		Then, we use the algorithm for scheduling jobs with resource demands equal to or smaller than 0.5 of the server's capacity in conjunction with the classic weighted shortest processing time algorithm for jobs with resource demands higher than 0.5. We thus create a hybrid, constant approximation algorithm for two or more machines. We also develop a constant approximation algorithm for the case with a single machine. This research is the first, to the best of our knowledge, to propose a polynomial-time algorithm with a constant approximation ratio for minimizing the weighted sum of job completion times for capacitated parallel machines.
	\end{abstract}

	
	
	%
	
	
	\section{Introduction}
	In this work, we study capacitated machine scheduling problems. These problems were initially encountered in production settings where jobs are processed in batches (e.g., scheduling jobs for heat treatment ovens and wafer fabrication processes \cite{damodaran2013grasp}).
	Recently, interest in these problems has increased because of the search for solutions that can be used in modeling modern cloud computing environments  \cite{fox2013weighted}. Differently than most scheduling models in which a resource serves a single job at any given time \cite[e.g.,][]{graves1981review,pinedo2012scheduling,balouka2019robust,cohen2021adaptive}, in modern cloud computing environments, multiple jobs can run concurrently on the same server subject to its capacity constraints (e.g., memory, cores, bandwidth). This fact can be seen in the following three examples of well-known cloud computing platforms in which resources are simultaneously shared by multiple jobs and clients: \gls{AWS}, Microsoft Azure and \gls{GCP}. These platforms leverage virtualization technologies, such as VMware products and Xen, to allow each physical machine to be shared by multiple jobs. Virtualization also helps in reducing the costs of maintenance, operation and provisioning \cite{malhotra2014virtualization}. 
	
	Managers of cloud computing environments who wish to increase the utilization of their data centers typically resort to improving the scheduling algorithms that allocate jobs to machines. In this context, we focus on minimizing the weighted sum of job completion times, which is one of most common objective functions \cite{fox2013weighted}. The weights imply that some jobs may be more important than others; thus, the scheduler has to take into consideration that the delay of one job can incur a higher ``cost'' than the delay of another.  
	
	The non-capacitated counterparts of the considered problem have been widely researched. One of the most known results is that the \gls{SPT} priority rule minimizes the (non-weighted) sum of job completion times \cite{pinedo2012scheduling}. The \gls{SPT} was extended to the \gls{WSPT}, which was used in the weighted version of the completion time minimization problem. Since the latter problem is $\mathcal{NP}$-complete for more than two machines \cite{garey1979computers}, various solution approaches focus on developing polynomial-time heuristic approximation algorithms that bound the worst case performance with respect to an optimal solution. For concreteness, a desired algorithm $A$ for a minimization problem $P$ would have a $\rho$ approximation ratio such that for any instance $I$ of $P$, $A(I) \leq \rho \cdot OPT(I)$, where $OPT(I)$ is the value of an optimal solution for $I$.  \citeauthor{eastman1964bounds} in \cite{eastman1964bounds} proved that scheduling according to the \gls{WSPT} priority rule provides a constant approximation algorithm.  \citeauthor{kawaguchi1986worst} \cite{kawaguchi1986worst} improved the \gls{WSPT} approximation ratio ($\rho$) to $(1+\sqrt{2}) / 2$ and proved that it is tight.
	
	In the capacitated setting, \citeauthor{im2016scheduling} \cite{im2016scheduling} were the first to develop a constant approximation algorithm for minimizing the \textit{non-weighted} sum of completion times using the \gls{SVF} priority rule combined with the \gls{SPT} priority rule. Our work extends theirs into the more general, weighted case, by combining \gls{WSVF} and \gls{WSPT} with the objective of developing the first algorithm with a constant approximation ratio for the weighted completion time problem in the capacitated setting. Our analysis, moreover, improves their approximation ratio for the non-weighted case.
	
	\subsection{Contributions and Techniques}
	We develop a constant approximation ratio algorithm for minimizing the weighted sum of job completion times on capacitated machines. Here are the main contributions: 
	\begin{enumerate} 
		\item A polynomial-time scheduling algorithm with a $\left(1+\frac{1}{1-\alpha}\right)$-approximation ratio, if the ratio between jobs' demands and the servers' capacities is at most $\alpha$.
		\item A polynomial-time scheduling algorithm with a $4+o(1/M)$ approximation ratio for $M \geq 2$ machines. 
		\item A polynomial-time scheduling algorithm with a $12+\epsilon$ approximation ratio for a single machine ($M$=1). 
		
	\end{enumerate}
	To the best of our knowledge, our result is the first constant approximation algorithm for the weighted completion time minimization problem for the capacitated parallel machine problem. In addition, we improve the approximation guarantees for the non-weighted version. In~\cite{im2016scheduling}, the authors proved a $\left(\frac{3\alpha}{1-\alpha} + 3\right)$-approximation ratio, for the case when the ratio of jobs' demands and servers' capacities is at most $\alpha$, and a $\left(5+o(1/M)\right)$ approximation ratio for $M>1$.
	Our algorithm for $M>1$ partitions the jobs into high- and low-resource demand classes where jobs within the former class require $50\% $ or more of a machine's capacity and jobs in the latter class require less than $50\%$. Accordingly, the algorithm partitions the machines into two groups for processing the two job classes. The high-demand class is scheduled via the \gls{WSPT} and the low-demand class via a \gls{WSVF} priority rule. In the \gls{WSVF} method, jobs are ordered in a non-decreasing order of the ratio between their processing time multiplied by the demand to the weight. Then, the jobs are assigned to a machine according to their priority, with the algorithm assigning the next unscheduled job to the earliest possible time $t$. Our proof involves analyzing the \gls{WSVF} performance for instances where job resource demands are smaller than a constant $\alpha$.
	
	In the analysis of the \gls{WSVF}, we bound the start time of each job, provided that prior to the start time, all the machines processed at least $1-\alpha$ demand of higher priority jobs. Then, 
	we bound the optimal cost, by the optimal cost of a non-capacitated converted instance.
	Finally, we develop an improved bound by using the characterization of \citeauthor{eastman1964bounds} \cite{eastman1964bounds} for the non-capacitated setting.
	For the single machine case, we extend the algorithm of \citeauthor{im2016scheduling} \cite{im2016scheduling} to the weighted case (the full details for this case are presented in the appendix).
	
	\subsection{Prior Work}
	There is a vast amount of research about machine scheduling problems owing to their theoretical and practical importance (interested readers are referred to the reviews by \citeauthor{meiswinkel2018mechanism} \cite{meiswinkel2018mechanism} and \citeauthor{liu2020review} \cite{liu2020review}). For the sake of brevity, we focus on recent capacitated machine scheduling studies. 
	
	Researchers explored several objective functions. One very popular objective is to minimize the processing makespan---that is, the completion time of the last job (e.g., \citeauthor{muter2020exact} \cite{muter2020exact} and \citeauthor{matin2017makespan} \cite{matin2017makespan}).
	Others suggested that in settings with a release time and deadline for each job, a reasonable objective is to maximize the total weight of jobs completed before their deadline (see \cite{albagli2014scheduling} and \cite{guo2017efficient}). Another line of research models the capacitated machine scheduling problem as an online problem in which jobs arrive over time \cite{kumar2019comprehensive}. In this line of research, typical objective functions are to minimize the response time (i.e., the time elapsed from the job's arrival until it is scheduled) and to maximize the throughput \cite {hota2019survey}.
	
	Three works that considered the weighted sum of flow-times or completion times are \cite{fox2013weighted}, \cite{im2016scheduling} and \cite{liu2019online}. \citeauthor{fox2013weighted} \cite{fox2013weighted} considered an online problem of weighted flow time minimization, assuming that jobs can be preempted with no penalty and delay. It is important to note, though, that preemption
	may incur significant switching costs (e.g., setup costs) and memory loss; preemption may be also forbidden due to system restrictions or client commitments.
	
	\citeauthor{liu2019online} \cite{liu2019online}, who also studied an online capacitated machine scheduling problem, assumed that a job can run at a slower rate when receiving a fraction of its demand or that a job can be processed in parallel on different machines. They then used \gls{OCO} to solve the scheduling optimization problem.  These assumptions may hold in specialized computing environments but in standard environments it may be costly or technically infeasible to split a job between machines or to process it at a slower rate using a portion of the required resources. 
	
	\citeauthor{im2016scheduling} \cite{im2016scheduling} proposed a constant approximation algorithm for minimizing the sum of completion times. They, however, did not consider the more general weighted version of the problem. We close this gap by developing an approximation algorithm that solves the weighted version of the problem and improves the approximation ratio for the non-weighted version that was presented in \cite{im2016scheduling}. 
	
	\section{Formal Problem Definition}
	We consider $N$ jobs that need to be processed by $M$ identical machines. Each job $j$ has a processing time $p_j$, demand $d_j$ and weight $w_j$, which are known in advance.
	We assume, without loss of generality (hereafter, w.l.o.g.), that $p_j\geq 1$,  and  $d_j \in (0,1]$ is a fraction of the required demand with respect to a machine's capacity. We denote $v_j = p_j \cdot d_j$ as job $j$'s volume.
	We focus on a non-preemptive schedule, meaning that a started job is processed without interruption until its completion. 
	The scheduler assigns each job to a machine and determines its start time, $s_j$; accordingly, the completion time of the job is $c_j = s_j + p_j$.
	
	Let $G^i(t)$ be the set of jobs processed by machine $i$ at time $t\in T$, where $T$ is an upper bound on the overall processing time. $j \in G_i(t)$ if job $j$ is assigned to machine $i$ and $t\in [s_j,c_j)$.
	A feasible schedule must ensure that the total demand of the jobs assigned to a machine does not exceed its capacity, at any given time. Mathematically,
	\begin{equation}\label{eq:capacity}
		\sum\limits_{j \in G^i(t)} d_{j} \leq 1 \; \; \; \forall i \in M , t \in T.
	\end{equation}
	
	Our goal is to find a feasible solution that minimizes the weighted sum of completion times:
	\begin{center}$ \min \sum \limits_{j=1}^N w_j \cdot c_{j}. $\end{center}
	
	As mentioned, the problem is $\mathcal{NP}$-complete. Accordingly, we are looking for a polynomial-time scheduling algorithm with a guarantee on the maximal ratio between the objective function value achieved by the algorithm and the optimal solution value. We want the developed approximation algorithm to provide a constant ratio. As discussed next, we base our algorithm on the \gls{WSVF} priority rule.
	
	\section{WSVF Algorithm and Analysis}
	The \gls{WSVF} algorithm orders the jobs in a non-decreasing order according to their volume over weight values, i.e., $(p_j \cdot d_j) / w_j$. The algorithm schedules the highest priority unassigned job (the one with the smallest value) at the earliest time $t$ on a machine that is available to process the job until it is completed; see Figure~\ref{fig:wsvf} for an example.
	

	\begin{figure}[H]
		\parbox{0.5\textwidth}{
			\centering
			\includegraphics[width=0.8\linewidth]{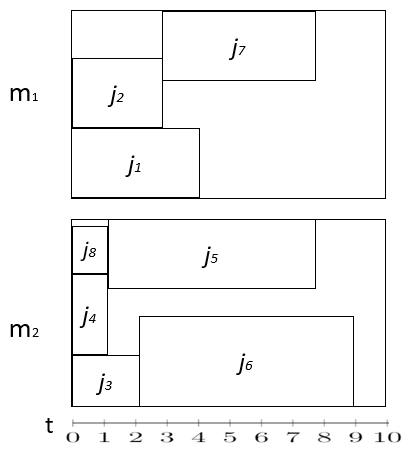}
			\label{fig:sub1}
		}
		\parbox{0.5\textwidth}{
			\centering
			\begin{tabular}{c c c c c}
				\hline
				jobs & $p_j$ & $d_j$ & $w_j$ & $\frac{p_j\cdot d_j}{w_j}$\\ [0.5ex]
				\hline
				1 & 4 & 0.4 & 8 & 0.2\\
				\hline
				2 & 3 & 0.4 & 5 & 0.24\\
				\hline
				3 & 2 & 0.25 & 1.5 & 0.33\\
				\hline
				4 & 1 & 0.45 & 1 & 0.45\\
				\hline
				5 & 7 & 0.4 & 4 & 0.7\\
				\hline
				6 & 7 & 0.5 & 4& 0.875\\
				\hline
				7 & 5 & 0.45 & 2 & 1.125\\
				\hline
				8 & 1 & 0.28 & 0.2 & 1.4\\
				\hline
			\end{tabular}
			\label{fig:sub2}
		}
		\caption{An illustration of the \gls{WSVF} algorithm's run. In this illustration we have $M=2$ machines and $N=8$ jobs. Each job is defined by its parameters ($p_j,d_j,w_j$) as described above. The jobs are sorted by a non-decreasing order of $p_j\cdot d_j / w_j$, for example, $p_1\cdot d_1 / w_1 = 0.2 < 0.24 = p_2 \cdot d_2 / w_2$. Then running over the jobs according to the priority order, each is scheduled on the machine with the earliest available $d_j$ capacity for $p_j$ time steps. Here, jobs $j_1,j_2,j_3$ and $j_4$ are scheduled at time $0$ according to the available machine capacities. Job $j_5$ can be scheduled only after the completion of $j_4$ so $s_5=c_4=1$, and jobs $j_6$ and $j_7$ are waiting for the completion of jobs $j_3$ and $j_2$, respectively. Job $j_8$, although it has a low priority, is scheduled at time $0$ because machine $m_2$ has a resource 'window' that is large enough to process job $j_8$. The objective function value in the example is $\sum \limits_{j=1}^8 w_j \cdot c_j = 8\cdot4+5\cdot3+ \dots + 2 \cdot 8 + 0.2 \cdot 1 = 135.2$.}
		\label{fig:wsvf}
	\end{figure}

	\begin{algorithm}[H]\label{alg:wsvf}
		\caption{\gls{WSVF}}
		\begin{algorithmic}
			
			\STATE Set $\mathcal{J}$ to be list of jobs, sorted in a non-decreasing order of their $(p_j \cdot d_j) / w_j$ values
			\FOR { job $j \in \mathcal{J}$}
			\STATE Schedule job $j$ on machine $i$ that provides the earliest start time $t$, s.t. 
			$$ d_j +  \displaystyle\sum_{h \in G_i(t')} d_{h} \leq 1 \text{, for all } t' \in [t, t+p_j)$$
			\ENDFOR
		\end{algorithmic}
	\end{algorithm}
	
	By definition, the assignment of Algorithm \ref{alg:wsvf} is feasible. Next, we show that its approximation ratio depends on $\alpha=\max_{j \in N}d_j$, the maximum resource demand of any $j \in N$.
	
	\begin{theorem}\label{theor:WSVF}
		If any job requires at most $\alpha < 1$, \gls{WSVF} is a $\left(\frac{1}{1-\alpha}+1\right)$-approximation algorithm for the weighted completion time minimization problem.
	\end{theorem}
	
	To prove Theorem~\ref{theor:WSVF}, we need to establish bounds on a problem instance $\hat{I}$, which is a compressed instance of the original problem instance ${I}$, where $I$ denotes the set of jobs $\mathcal{J}$ ordered by the \gls{WSVF} priority rule. We ``compress'' the original jobs such that their demands become $1$ (as in the non-capacitated setting). For this, we take the ordered set of jobs $I$, and for each job in the set, we fix the duration to be $\hat{p}_j= p_j \cdot d_j$, which is the so-called volume of job $j$ in the original instance, and fix the demand as $\hat{d}_j = 1$. 
	We note that the compressed instance preserves the ordering of the jobs and the list of jobs is ordered according to the \gls{WSPT} priority rule. 
	
	We bound the cost of \gls{WSVF} on $I$ by proving that the optimal solution value of the corresponding compressed instance $\hat{I}$ is smaller than or equal to the optimal solution value of $I$. 
	For ease of notation, we denote by $C_M(I)$ the objective function value (cost) achieved by algorithm \gls{WSVF} for instance $I$ using $M$ machines and by $C^*_M(I)$ the optimal cost of scheduling $I$ on $M$ machines.
	
	We note the following two observations: 1) $C_N(I) = \sum_{j\in I} w_j p_j$ since when using $N$ machines and $N=|I|$, it is always optimal to assign each job to a machine, and 2) for a single machine, $C_1(\hat{I}) = C^*_1(\hat{I}) = \sum_j (w_j \cdot \sum\limits_{h=1}^{j}\hat{p}_h)$ by Smith's rule~\cite{smith1956various}.

	We begin by establishing the relations between $C^*_M(I), C^*_M(\hat{I}),C^*_1(\hat{I})$ and $C_M(I)$.
	First, since $c_j \geq p_j$ always holds:
	
	\begin{observation}\label{obv:cn}
		$C_N(I) \leq C^*_M(I)$.
	\end{observation}
	
	Next, we bound the optimal cost on $M$ machines by the optimal cost on the compressed instance on $M$ machines.
	\begin{claim}\label{clm:opt_star}
		$C^*_M(\hat{I}) \leq C^*_M(I)$.
	\end{claim}
	\begin{proof}
		Consider, w.l.o.g., one of the $M$ machines and let $j_{1}, j_{2} ,...., j_{k}$ be the jobs assigned to this machine ordered by their completion times in the optimal schedule (which has a cost $C^*_M(I)$).
		We prove that a scheduler that processes the compressed jobs on the same machine in this order has a smaller than or equal to cost compared to the uncompressed optimal instance.
		Let $c^*_j$ be the completion time of job $j$ in the optimal original instance, and $\hat{c}_j$ be its completion time in the compressed instance schedule. We will show that $c^*_j\geq \hat{c}_j$.
		First, observe that $\hat{c}_{j} = \sum_{h=1}^j \hat{p}_h$, since in the compressed schedule, the machine can process a single job at each time. 
		Next, by definition, at time $c^*_j$, all the first $j$ jobs have been completed on this machine. By the volume preservation rule, it took at least $\sum_{h=1}^j {p_{h}\cdot d_{h}}$ time units to complete the jobs since at any time step, the machine can process a maximal volume of $1$; therefore, $c^*_j \geq \sum_{h=1}^j {p_{h}\cdot d_{h}}  = \sum_{h=1}^j \hat{p}_h = \hat{c}_{j}$.
		Hence, when we sum over all machines and jobs,
		$$C^*_M(\hat{I}) \leq \sum\limits_{j \in I} w_j\cdot  \hat{c}_j \leq \sum\limits_{j \in I} w_j\cdot  c^*_j = C^*_M(I).$$
		The \gls{LHS} inequality holds since the optimal cost $C^*_M$ is equal to or smaller than the cost of any schedule and the \gls{RHS} inequality follows from extending the volume preservation argument to $M$ machines.
	\end{proof}
	
	Next, we use Theorem 1 from \cite{eastman1964bounds} to construct a lower bound for $C^*_M(\hat{I})$ using a single machine optimal cost $C^*_1(\hat{I})$.
	
	\begin{claim} \label{clm:1_fast}
		$C^*_M(\hat{I})\geq \frac{1}{M}C^*_1(\hat{I}).$
	\end{claim}
	\begin{proof}
		\citeauthor{eastman1964bounds} \cite{eastman1964bounds} proved a tight bound on the optimal cost of the non-capacitated case with $M$ machines with respect to the optimal cost of an instance with a single machine. Specifically, they proved 
		$$C^*_M(\hat{I})-\frac{1}{2}C_N(\hat{I})\geq \frac{1}{M}(C^*_1(\hat{I}) - \frac{1}{2}C_N(\hat{I})).$$
		Therefore, we have  
		$$C^*_M(\hat{I})\geq \frac{1}{M}C^*_1(\hat{I}) + \frac{1}{2}C_N(\hat{I})(1-\frac{1}{M})\geq \frac{1}{M}C^*_1(\hat{I})$$ 
		since $\frac{1}{2}C_N(\hat{I})(1-\frac{1}{M})\geq 0$.
		
	\end{proof}
	We now compare $C_M(I)$ and $C_1(\hat{I})$ by bounding the start time $s_j$ for each job $j\in N$. We prove that $s_j$ can be bounded by a factor that depends on $\alpha$ and $M$ times the total volume of jobs preceding $j$ in $I$ (which is ordered by \gls{WSVF}). In other words, $\sum_{h<j} v_h$, which is the start time of job $j$ in $C_1(\hat{I})$.  
	
	\begin{lemma} \label{lem:start_time}
		$\forall j \in N: s_{j} \leq \frac{1}{(1-\alpha)M} \Sigma_{h=1}^{j-1} v_{h}$
	\end{lemma}
	
	\begin{proof}

		To prove Lemma~\ref{lem:start_time}, we prove that prior to $s_j$, each machine processes a total demand of at least $1-\alpha$ on jobs with a higher priority than $j$. 
		This property also appeared in \citeauthor{im2016scheduling} \cite{im2016scheduling} and holds for any priority-based algorithm. For a machine $i\in M$ at time $t$ and job $j\in I$, we define  $D^i_{j}(t) = \sum_{j'\in G^i(t),h\leq j} d_{h}$, which is the total demand of jobs
		among the first $j$ jobs that are processed on machine $i$ at time $t$.
		By proving the following property, we prove that for all $t < s_{j}$, we have $D^{i}_{j-1}(t) \geq 1- \alpha$.

		\begin{claim} \label{clm:invariant} 
			For all $i \in M, j \in N$,  $min\{ D^i_j(t), 1-\alpha \}$ is non-increasing with $t$.
		\end{claim}
		
		\begin{proof}
			Choose a machine $i$ (we omit the superscript $i$ when it is clear from the context), and consider w.l.o.g. only jobs assigned to this machine (with the same priority order as in $I$).
			We prove the claim using an induction on $j$. First, consider the case of $j=1$,
			
			$$ \min\{ D_1(t), 1-\alpha \} = \begin{cases} \min \{d_1, 1-\alpha \}&\mbox{if } t \leq p_{1},\\
				0 & \mbox{if } t>p_{1}. \end{cases}$$
			
			It is a non-increasing function with $t$ so the claim holds for $j=1$.
			
			Now we use the induction assumption that the claim is true for $j-1$ and prove the claim for $j>1$. 
			First, we argue that according to the \gls{WSVF}, $D_j(t) = D_{j-1}(t) > 1 - \alpha$, for all $0 \leq t< s_j$. Otherwise, if  $0 \leq t<s_j$ exists such that  $D_{j-1}(t) \leq 1-\alpha$ (that is, $D_{j-1}(t)+\alpha \leq 1$), then for $t\leq t'<s_j$, we have $D_{j-1}(t')\leq 1-\alpha$ by our induction assumption, and since $d_j \leq \alpha$ by our assumption on the input, the \gls{WSVF} would assign job $j$ before time $s_j$. 
			
			Now, we can look at the expression $min \{ D_{j}(t), 1- \alpha \}$ value over time $t$:\\
			$$\min \{ D_{j}(t), 1- \alpha \} = \begin{cases} 1-\alpha &\mbox{if } t < s_{j},\\
				\min \{ D_{j-1}(t)+d_{j}, 1- \alpha \}& \mbox{if } s_{j} \leq t \leq s_{j}+ p_{j}, \\
				\min \{ D_{j-1}(t), 1- \alpha \} & \mbox{if } t > s_{j}+ p_{j}. 
			\end{cases}$$
			Using the above and the induction assumption, we can conclude that $min \{D_{j}(t),1-\alpha\}$ is non-increasing with $t$.
		\end{proof}
		
		\begin{corollary}
			\label{col:usage} 
			For all $i \in M, j \in N$  $0\leq t < s_{j}$, we have  $ D^{i}_{j-1}(t) > 1- \alpha$.
		\end{corollary} 
		
		\begin{proof}
			
			We follow the proof of Claim \ref{clm:invariant}; if at any $t<s_j$, $D_{j-1}(t) \leq 1-\alpha$, then by the invariant of Claim~\ref{clm:invariant}, the \gls{WSVF} would process job $j$ earlier than $s_j$ by its definition.
		\end{proof}
		
		Next, we use Corollary~\ref{col:usage} to set an upper bound on the starting time of every job $j$ and to conclude the proof of Lemma \ref{lem:start_time}. By corollary \ref{col:usage}, we have that each machine used at least $(1-\alpha)$ of its capacity until time $s_{j}$, to process the first $j-1$ jobs.
		Thus, the following inequality holds for $M$ machines
		$$\Sigma_{h=1}^{j-1} v_{h} \geq M(1- \alpha)s_{j},$$
		where the \gls{LHS} is the sum of the first $j-1$ job volumes that were processed on $M$ machines and the \gls{RHS} follows from the lower bound on the used volume extended to $M$ machines. Reorganizing the above formula leads to an upper bound on the starting time of job $j$:
		$$ s_{j} \leq \frac{1}{(1-\alpha)M}  \Sigma_{h=1}^{j-1} v_{h}.$$
	\end{proof}
	
	By using the above results, we can prove the main lemma for bounding $C_M(I)$: 
	\begin{lemma}\label{lem:WSVFmain}
		Given any instance $I$, such that for all $j \in I$, $d_j\leq \alpha$ for $0<\alpha<1$ we have:
		$$C_M(I) \leq C_N(I) + \frac{C_1(\hat{I})}{(1-\alpha)M}$$
	\end{lemma}
	
	\begin{proof}
		For every job $j\in N$
		\begin{equation} \label{eq:singm}
			c_j = p_j + s_j\leq p_j + \frac{1}{(1-\alpha)M}\cdot \sum\limits_{h=1}^{j-1}v_h =p_j + \frac{1}{(1-\alpha)M}\cdot\sum\limits_{h=1}^{j-1}\hat{p}_h,    
		\end{equation}
		where the \gls{LHS} equality exists by definition, the inequality follows from Lemma \ref{lem:start_time}, and the \gls{RHS} equality follows from the definition $v_h= d_h\cdot p_h= \hat{p}_h$.
		
		Summing over all jobs to find the total cost, we have:
		\begin{eqnarray*}
			C_M(I) &=& \sum_{j=1}^{N} w_j \cdot c_j \\ &\leq& \sum_{j=1}^{N}  w_j \cdot p_j + \sum_{j=1}^{N}  w_j \cdot \left(\frac{1}{(1-\alpha)M}\cdot\sum_{h=1}^{j-1}\hat{p}_h\right) \\ &=& \sum\limits_{j=1}^{N}  w_j \cdot p_j + \frac{1}{(1-\alpha)M} \cdot \sum\limits_{j=1}^{N} (w_j \cdot \sum\limits_{h=1}^{j-1}\hat{p}_h) \\
			&\leq& C_N(I) + \frac{C_1(\hat{I})}{(1-\alpha)M}, 
		\end{eqnarray*}
		
		where the first inequality follows from Equation~\ref{eq:singm} and the second inequality holds since
		$C_N(I) = \sum\limits_{j=1}^{n}  w_j \cdot p_j$, and $C_1(\hat{I}) = \sum\limits_{j=1}^{n} (w_j \cdot \sum\limits_{h=1}^{j}\hat{p}_h)\geq \sum\limits_{j=1}^{n} (w_j \cdot \sum\limits_{h=1}^{j-1}\hat{p}_h)$.
	\end{proof}
	
	Finally, we prove Theorem~\ref{theor:WSVF}, which follows immediately from Lemma~\ref{lem:WSVFmain} and Claims~\ref{clm:opt_star} and \ref{clm:1_fast}.
	\begin{proofof}{Theorem~\ref{theor:WSVF}}
		$$C_M(I) \leq C_N(I) + \frac{C_1(\hat{I})}{(1-\alpha)M}\leq  C_N(I) + \frac{C^*_M(\hat{I})}{1-\alpha}\leq  C^*_M(I) \left(\frac{1}{1-\alpha}+1\right).$$
	\end{proofof}
	
	\section{Constant Approximation Algorithm for all Instances}
	In this section, we extend the previous results to remove the dependency on $\alpha$, the maximum demand of any job, which leads to a constant approximation algorithm. We rely on the observations that, given a set of jobs with demands that are larger than $1/2$, only a single job can run on a machine at any time, and that the \gls{WSVF} has a constant approximation ratio on jobs with demands equal to or smaller than $1/2$. Thus, by splitting the jobs based on their resource demands into two  sets of machines, we achieve a constant approximation algorithm.
	
	To this end, we introduce the Hybrid-\gls{WSVF}:
	
	\begin{algorithm}[H]
		\label{alg:hybrid}
		\SetAlgoLined
		\begin{algorithmic}
			\STATE Split the jobs in $I$ into two sets $I^l = \{j: d_j \leq \frac{1}{2}\}$ and $I^h = \{j: d_j > \frac{1}{2}\}$. 
			
			\STATE Schedule $I^l$ on $M_1 = \lceil \frac{2(M-2)}{3}\rceil +1$ machines using \gls{WSVF}.
			\STATE Schedule $I^h$ on $M_2=M-M_1$ using \gls{WSPT}. 
		\end{algorithmic}
		
		\caption{HYBRID-\gls{WSVF}}
	\end{algorithm}
	
	\begin{theorem}\label{theor:H-WSVF}
		HYBRID-WSVF is a $4+o(\frac{1}{M})$-approximation for the weighted completion time minimization problem where $M$ is the number of machines.
	\end{theorem}
	
	\begin{proof}
		First, we schedule the low demand jobs (i.e., $d_j\leq 0.5$) on $M_1$ machines.
		Thus, from Lemma~\ref{lem:WSVFmain}, we have
		\begin{eqnarray*}
			C_{M_1}(I^l) &\leq& C_N(I^l) + \frac{1}{(1-\alpha)M_1} \cdot C_1(\hat{I}^{l}). 
		\end{eqnarray*}
		For the low demand jobs, $\alpha$ values are within the interval $(0, 0.5]$ with a specific value per  $I^l$. Thus,
		\begin{eqnarray*}
			C_N(I^l) + \frac{1}{(1-\alpha)M_1} \cdot C_1(\hat{I}^{l}) &\leq& C_N(I^l) + \frac{2}{M_1} \cdot C_1(\hat{I}^{l}) \\ 
			&=& C_N(I^l) + \frac{2 M}{M_1} \cdot \frac{C_1(\hat{I}^{l})}{M} \leq C^*_M(I^l) + 2 \cdot \frac{M}{M_1} C^*_M(I^l),
		\end{eqnarray*}
		where the last \gls{RHS} inequality follows from Observation \ref{obv:cn}, Claim \ref{clm:opt_star} and Claim \ref{clm:1_fast}.

		Next, we deal with the high demand jobs ($d_j>0.5$), which are scheduled on $M_2$ machines. Since any two high demand jobs cannot run simultaneously on the same machine, the scheduling problem is equivalent to the weighted non-capacitated setting. Therefore, we can use a result from  \citeauthor{eastman1964bounds} \cite{eastman1964bounds} who proved the following bound for the \gls{WSPT} with $M$ machines:
		$$ C_M(\hat{I}) \leq C_N(\hat{I}) + \frac{1}{M}C_1(\hat{I}).$$
		Following their result, we state that:
		$$ C_{M_2}(I^h) \leq C_N(I^h) + \frac{1}{M_2}C_1(I^h) = C_N(I^h) + \frac{M}{M_2}\frac{C_1(I^h)}{M}\leq C^*_M(I^h) + \frac{M}{M_2}C^*_M(I^h),$$
		
		\noindent where the \gls{RHS} inequality follows from Observation \ref{obv:cn} and Claim~\ref{clm:1_fast}.
		
		By setting $M_1 = \lceil \frac{2(M-2)}{3}\rceil +1$ and $M_2 = M-M_1=\lfloor\frac{M-2}{3}\rfloor+1$, and by 
		observing $C^*_M(I) \geq C^*_M(I^h)+C^*_M(I^l)$, we conclude the proof for Theorem \ref{theor:H-WSVF}.
	\end{proof}
	
	\section{Conclusions and Future Work}
	Algorithms for scheduling jobs on capacitated machines can significantly affect the performance of cloud computing environments. While such environments handle jobs of varying importance (e.g., cost), there are no constant approximation algorithms, to the best of our knowledge, for the related problem of minimizing the weighted sum of job completion times. This paper closes this gap. The suggested algorithm also improves the best-known approximation ratio for the non-weighted problem (for $M \ge 2$). 
	
	The results presented in this paper may be enhanced by going in several research directions, of which we mention three. The first significant theoretical extension may consider capacitated problems with multiple capacitated resources. In real-world cloud computing environments, CPU, memory, storage and bandwidth may be scarce resources. To realize such an extension, one would need to consider a multidimensional demand for each job. 
	A second extension would be to use our methods to find performance guarantees in stochastic environments in which, for example, processing duration can be characterized using probabilistic knowledge. The third research direction would be to develop a model that accommodates jobs with release dates and deadlines. Such research may  also be considered to be a natural extension of the current model. 
	\printbibliography
	\appendix
	\section*{Appendix: The single machine case}
	In this section, we introduce a different algorithm for the special case of $M=1$, and conclude that for the weighted completion time minimization for capacitated machine, there exists a constant competitive algorithm for any number of machines.  Note that Algorithm \ref{alg:hybrid} is well defined for $M \geq 2$ machines; for $M=1$ machines, the partition approach cannot be applied.
	Instead, we extend the ideas of the non-weighted case for $M=1$ in \cite{im2016scheduling} to the general weighted case.
	
	The algorithm uses two main tools, a knapsack algorithm that determines a subset of highest weight jobs with a total volume restriction, and a 2D-strip packing algorithm for scheduling this set, where each job $j$ is represented by a rectangle $r_j$ defined with width $p_j$ and height $d_j$.
	
	
	The algorithm will work in iterations until it assigns all jobs, where in iteration $\ell$ the algorithm will execute the following steps:
	\begin{enumerate}
		\item Solve a knapsack packing problem: Compute $J_\ell$, a maximum weighted set of jobs with a total volume less than $2^\ell$ using a $1+\epsilon$ resource augmentation knapsack algorithm.
		\item Pack the jobs in $J_\ell$ into a $3$ strips of width $(1 + \epsilon) 2^\ell$ and height $1$
		\item Concatenate the strips from the earlier step and schedule the job on the machine. 
	\end{enumerate}
	
	Note that after $\ell^{\max} = \log(\sum_{j\in I} v_j)$ iterations, all the jobs are packed, and  the algorithm can discard jobs that were scheduled in previous iterations. We will show that the first two steps of finding the set of jobs and packing it into strips can be done in polynomial time. Finding a maximum weighted set of jobs ($\arg\max_J\{\sum_{j \in J}w_j: \sum_{j\in J}v_j \leq B \}$) is equivalent to the knapsack problem: 
	Given a set of items, each with a size and a value, determine which item to include in a collection so that the total size is less than or equal to a given limit and the total value is as large as possible.
	We apply the resource augmentation solution of ~\cite{williamson2011design} for this problem, which proves that given a bound on the total size $L$, there exists a polynomial time in which  a set with of a total size $(1+\epsilon)\cdot L$ with a total profit of at least the optimal profit of a set with total size $L$ can be computed.
	Second, we utilize \citeauthor{jansen2007maximizing}'s algorithm~\cite{jansen2007maximizing}, which packs rectangles (without rotations) with total volume of  $L$ and maximal height of $1$ into a $3$ strips of width $L$ and height $1$.
	
	\begin{algorithm}[H]
		\label{alg:one_machine}
		\SetAlgoLined
		\begin{algorithmic}
			\FOR {$\ell = 0,1, 2, \dots, \ell^{\max}$}
			\STATE Compute $J_\ell\subseteq I$, a maximal weighted set with total volume of at most $2^\ell \cdot (1+\epsilon)$. 
			\STATE Pack $J_\ell$ into strip $S$ of height $1$ and width $3\cdot(1+\epsilon)\cdot 2^\ell$.
			\STATE  Schedule $S$ in the interval $[3(1+\epsilon) \sum \limits_{h=0}^{\ell-1}2^h,3(1+\epsilon) \sum \limits_{h=0}^{\ell}2^h)$
			\ENDFOR
		\end{algorithmic}
		
		\caption{PackAndSchedule($I$)}
	\end{algorithm}
	
	\begin{theorem}\label{theor:one_machine}
		PackAndSchedule is a $12(1+\epsilon)$ approximation polynomial-time algorithm for total weighted completion time minimization in the single machine case.
	\end{theorem}
	\begin{proof}
		
		The algorithm provides a feasible scheduling allocation since each of the 2D packing solution strips having a maximal height (demand) of $1$ is assigned to a disjoint time interval.
		
		Let $W_\ell = \sum_{j\in J_\ell} w_j$ be the sum of weights that are processed in  iteration $\ell$. Note that jobs in $J_{\ell}$ are scheduled until time  $3(1+\epsilon)\cdot (2^{\ell+1}-1)$. In addition, when jobs in $J_{\ell}$ are processed, the total weight of jobs that are not yet completed is at most $W-W_{\ell-1}$.
		Therefore,  we can bound the cost of the algorithm \ref{alg:one_machine} by summing over the iterations and multiplying the unprocessed weight by the completion time:
		$$ C_1(I) \leq 3(1+\epsilon) \sum \limits_{\ell \geq 0} (2^{\ell+1}-1)(W-W_{\ell-1})$$
		The optimal scheduler can complete jobs with a total weight of at most $W_\ell$
		up to time $2^{\ell}$, because by volume preservation, it is not possible to pack more than $2^{\ell}$ volume until this point, and by the correctness of the knapsack algorithm, this is the maximum profit (weight) that can be packed by $W_\ell$ if the total volume is at most $2^{\ell}$; therefore, there are jobs with a total weight of at least $W-W_\ell$ that are unprocessed at time $2^{\ell}$. Hence, we can give a lower bound for the cost of the optimal scheduler:
		$$C_1^*(I) \geq W + \sum \limits_{\ell \geq 0} 2^{\ell}(W-W_\ell)$$
		
		
		%
		By combining the two above inequalities, we have that PackAndSchedule is a $12(1+\epsilon)$ approximation algorithm as required.
	\end{proof}

\end{document}